\documentclass[preprint]{elsarticle}
\journal{Artificial Intelligence}
\usepackage{times}
\usepackage{helvet}
\usepackage{courier}

\usepackage{amsthm,amsmath,amssymb,bbm,graphicx,rotating}
\usepackage{boxedminipage}
\usepackage{enumerate}
\usepackage{epic,eepic}
\usepackage{tikz}
\usetikzlibrary{shapes}
\usepackage[rightcaption]{sidecap}    
\usepackage{booktabs}
\usepackage{graphicx}
\usepackage{paralist}

\usepackage{floatrow}

\tikzstyle{every picture} = [>=latex]
\usepackage[author=,draft]{fixme}
\fxsetup{theme=color}

\usepackage{xspace}
\newcommand{\bigoh}{\mathcal{O}}

\renewcommand{\phi}{\varphi}

\makeatletter
\newcommand*{\ol}[1]{$\overline{\hbox{#1}}\m@th$}

\makeatother

\newcommand{\Nat}{{\mathbb{N}}}

\newcommand{\cF}{{\cal F}}

\newcommand{\var}{\text{\normalfont var}}

\newcommand{\SB}{\{\,}%
\newcommand{\SM}{\;{:}\;}%
\newcommand{\SE}{\,\}}%

\newcommand{\eval}{\eta}
\newcommand{\td}{td}
\newcommand{\tw}{tw}

\theoremstyle{plain}
\newtheorem{THE}{Theorem}
\newtheorem{PRO}[THE]{Proposition}
\newtheorem{LEM}[THE]{Lemma}
\newtheorem{COR}[THE]{Corollary}
\newtheorem*{THE*}{Theorem}
\newtheorem*{fact*}{Fact}
\theoremstyle{definition}
\newtheorem{DEF}[THE]{Definition}
\newtheorem*{DEF*}{Definition}
\theoremstyle{remark}
\newtheorem{OBS}{Observation}
\newtheorem{CLM}{Claim}

\newcommand{\pbDef}[3]{%
\noindent
\begin{center}
\begin{boxedminipage}{0.98 \columnwidth}
#1\\[5pt]
\begin{tabular}{l p{0.70 \columnwidth}}
Input: & #2\\
Question: & #3
\end{tabular}
\end{boxedminipage}
\end{center}
}

\newcommand{\cc}[1]{{\mbox{\textnormal{\textsf{#1}}}}\xspace}  

\newcommand{\integers}{\mathbb{Z}}
\renewcommand{\P}{\cc{P}}
\newcommand{\NP}{\cc{NP}}

\newcommand{\FPT}{\cc{FPT}}

\newcommand{\paraNP}{\cc{paraNP}}

\newcommand{\hy}{\hbox{-}\nobreak\hskip0pt}

\newcommand{\nn}{\mathbb{N}}

\newcommand{\bigO}[1]{\ensuremath{{\mathcal O}(#1)}}

\newcommand{\probfont}[1]{\textnormal{\textsc{#1}}}

\newcommand{\clos}{\ensuremath{\text{clos}}}
\begin{document}
\begin{frontmatter}
  \title{The Complexity Landscape of Decompositional Parameters for ILP}

  \cortext[cor1]{Corresponding authors}

  \author[ac]{Robert Ganian\corref{cor1}}
  \ead{rganian@gmail.com}
  \author[ac]{Sebastian Ordyniak}
  \ead{sordyniak@gmail.com}

  \address[ac]{Algorithms and Complexity Group, TU Wien, Favoritenstrasse 9-11, 1040 Wien, Austria}

\begin{abstract}
  Integer Linear Programming (ILP) can be seen as the
  archetypical problem for \NP-complete optimization problems, and a wide range of 
  problems in artificial intelligence are solved
  in practice via a translation to ILP. Despite its huge range of applications, only
  few tractable fragments of ILP are known, probably the most
  prominent of which is based on the notion of total
  unimodularity. Using entirely different techniques, we identify new
  tractable fragments of ILP by studying structural parameterizations
  of the constraint matrix within the
  framework of parameterized complexity. 

  In particular, we show
  that ILP is fixed-parameter tractable when parameterized by the treedepth
  of the constraint matrix and the maximum
  absolute value of any coefficient occurring in the ILP instance. Together
  with matching hardness results for the more general parameter treewidth, we
  give an overview of the complexity of ILP w.r.t.\ decompositional
  parameters defined on the constraint matrix.
\end{abstract}

  \begin{keyword}
    Integer Linear Programming, treewidth, treedepth, (Parameterized) complexity
  \end{keyword}
  
\end{frontmatter}

\section{Introduction}
\label{sec:intro}

Integer Linear Programming (ILP) is among the most successful and general
paradigms for solving computationally intractable optimization
problems in computer science. In
particular, a wide variety of problems in artificial intelligence are
efficiently solved in practice via a translation into an Integer
Linear Program, including problems from areas such as process scheduling~\cite{FloudasLin05}, planning~\cite{BrielVossenKambhampati05,VossenBallLotemNau99}, vehicle routing~\cite{Toth01}, packing~\cite{LodiMM02}, and network hub location~\cite{AlumurK08}. In its most
general form ILP can be formalized as follows:

\pbDef{\probfont{Integer Linear Program}}
{A matrix $A \in \integers^{m\times n}$ and two vectors $b \in
  \integers^m$ and $s \in \integers^n$.}{Maximize $s^Tx$ for every $x
  \in \integers^n$ with $Ax \leq b$.}

Closely related to ILP is the \probfont{ILP-feasibility} problem, where 
given $A$ and $b$ as above, the problem is to decide
whether there is an $x\in \integers^n$ such that $Ax\leq b$.
The decision version of ILP, \probfont{ILP-feasibility} and various other highly restricted
variants are well-known to be \NP-complete~\cite{Papadimitriou81}.

Despite the importance of the problem, an understanding of the influence of structural
restrictions on the complexity of ILP is still in its infancy. This is
in stark contrast to another well-known and general paradigm for the
solution of problems in Computer Science, the Satisfiability problem (SAT).
There, the parameterized complexity framework~\cite{book/DowneyF13} has yielded deep results capturing the tractability and intractability of SAT with respect to a plethora of structural restrictions. In the context of SAT, one often considers structural restrictions on a
graphical representation of the formula (such as the primal
graph), and the aim is to design efficient 
fixed-parameter algorithms for SAT, i.e., algorithms running in time
$\bigO{f(k)n^{\bigO{1}}}$ where $k$ is the value of the considered
structural parameter for the given SAT instance and $n$ is its input size. 
It is known that SAT is fixed-parameter tractable
w.r.t.\ a variety of structural parameters, including the prominent parameters
treewidth~\cite{Szeider03} but also more specialized parameters~\cite{FischerMakowskyRavve06,GanianSzeider15,GanianHlinenyObdrzalek10}.

\paragraph{Our contribution} In this work, we carry out a similar line of research for ILP by studying
the parameterized complexity of ILP w.r.t.\ various structural
parameterizations. In particular, we
consider parameterizations of the \emph{primal graph} of the
ILP instance, i.e., the undirected graph whose
vertex set is the set of variables of the ILP instance and whose edges
represent the occurrence of two variables in a common expression.
We obtain a complete picture of the parameterized complexity of ILP 
w.r.t.\ well-known decompositional parameters of the primal graph,
specifically treedepth, treewidth, and cliquewidth; our results are summarized in Table~\ref{tab:landscape}.

\begin{table}[b]
\begin{center}
\begin{tabular}{l|l|l}
~& $\ell$ &without $\ell$\\ \hline
TD & \FPT{} (Thm~\ref{thm:main}) & \paraNP{}\hy h (Thm~\ref{the:hard-treedepth})\\ 
TW/CW & \paraNP{}\hy h (Thm~\ref{the:para-hard-tw-coeff})
& \paraNP{}\hy h (Thm~\ref{the:para-hard-tw-coeff})\\ 
None & \paraNP{}\hy h (Obs~\ref{obs:hard-coef}) & n.a. \\ 
\end{tabular}
\end{center}
\vspace{-0.3cm}
\caption{The complexity landscape of ILP obtained in this paper. 
  The table shows the parameterized complexity of ILP parameterized by
  the treedepth (TD), treewidth (TW), or cliquewidth (CW) of the
  primal graph with (second column ``$\ell$'') and without (third
  column ``without $\ell$'') the additional parameterization by the
  maximum absolute value $\ell$ of any coefficient in $A$ or $b$.
}
\label{tab:landscape}
\end{table}
Our main algorithmic result (Theorem~\ref{thm:main}) shows that ILP is fixed-parameter tractable
parameterized by the \emph{treedepth} of the primal graph and the
maximum absolute value $\ell$ of any coefficient occurring in $A$ or
$b$. 
Together with the classical results for \emph{totally unimodular}
matrices~\cite[Section 13.2.]{PapadimitriouSteiglitz82} and fixed
number of variables~\cite{Lenstra83}, which use entirely
different techniques, our result is one of the few
tractability results for ILP without additional restrictions. We note that the presented algorithm is primarily of theoretical interest; the intent here is to classify the complexity of ILP by providing runtime guarantees, not to compete with state-of-the-art ILP solvers. 

We complement our algorithmic results with matching lower bounds, provided in terms
of \paraNP{}\hy hardness results~(see the Preliminaries); an overview of the obtained results is provided in Table~\ref{tab:landscape}. 
Namely, we show that already \probfont{ILP-feasibility} is unlikely to
be fixed-parameter tractable when parameterized by treedepth (whereas the case of parameterizing by only $\ell$ is known to be hard); in fact, our results also exclude algorithms running in time $(n+m)^{f(k)}$, where $k$ is the parameter.
Moreover, the hardness results provided here also hold in the strong sense, i.e., even for
ILP instances whose size is bounded by a polynomial of $n$ and $m$; it is worth noting that this requires a more careful approach than what would suffice for weak \paraNP{}\hy hardness.

One might be tempted to think
that, as is the case for SAT and numerous other problems, the fixed-parameter tractability
result for treedepth carries over to the more general structural
parameter treewidth. We show that this is not the case for ILP. Along with recent results for the Mixed Chinese Postman Problem~\cite{GutinJonesW15}, this is only the second known case of a natural problem where using treedepth instead of treewidth actually ``helps'' in terms of fixed parameter tractability.
In fact, we show that already \probfont{ILP-feasibility}
remains \NP{}\hy hard for ILP instances of treewidth at most two
and whose maximum coefficient is at most one. Observe that this also implies
the same intractability results for the more general parameter clique-width~\cite{CourcelleMakowskyRotics00}.

\paragraph{Related Work}

We are not the first to consider decompositional parameterizations of
the primal graph for ILP. However, previous results in this area
required either implicit or explicit bounds on the domain values of
variables together with further restrictions on the coefficients.
In particular, for the case of non-negative ILP instances, i.e.,
ILP instances where all coefficients as well as all variable domains
are assumed to be non-negative, ILP is known to be fixed-parameter tractable
parameterized by the branchwidth, a decompositional parameter closely
related to treewidth, of the primal graph and the maximum value $B$ of any
coefficient in the constraint vector $b$~\cite{CunninghamGeelen07}. Note that $B$ also bounds
the maximum domain value of any variable in the case of non-negative
ILP instances. A more recent result by Jansen and Kratsch~\cite{JansenKratsch15esa} showed
that ILP is fixed-parameter tractable parameterized by the treewidth
of the primal graph and the maximum absolute domain value of any
variable. Hence in both cases the maximum absolute
domain value of any variable is bounded by the considered parameters,
whereas the results presented in this paper do not require any bound
on the domain values of variables.

Furthermore, a series of tractability results for ILP based on
restrictions on the constraint matrix $A$, instead of restrictions on
the primal graph, have been
obtained~\cite{DeLoeraHK:AGideasBook,HemmeckeOR:13,Onn2010}. These
results apply whenever the constraint matrix $A$ can be written as an
arbitrary large product of matrices of bounded size and are usually
referred to as $n$-fold ILP, two-stage stochastic ILP, and 4-block
$n$-fold ILP.


\section{Preliminaries}
\label{sec:prelim}
We will use standard graph terminology, see for
instance~\cite{Diestel12}.
A graph $G$ is a tuple $(V,E)$, where $V$ or
$V(G)$ is the vertex set and $E$ or $E(G)$ is the edge set. 
A graph $H$ is a subgraph of a graph $G$, denoted $H\subseteq G$, if $H$ can be obtained by deleting vertices and edges from $G$.
All our graphs are simple and loopless. 

A \emph{path} from vertex $v_1$ to vertex $v_j$ in $G$ is a sequence of distinct vertices $v_1,\dots,v_j$ such that for each $1\leq i< j$, $\{v_i,v_{i+1}\}\in E(G)$.
A \emph{tree} is a graph in which, for any two vertices $v,w\in G$, there is precisely one unique path from $v$ to $w$; a tree is \emph{rooted} if it contains a specially designated vertex $r$, the \emph{root}. Given a vertex $v$ in a tree $G$ with root $r$, the \emph{parent} of $v$ is the unique vertex $w$ with the property that $\{v,w\}$ is the first edge on the path from $v$ to $r$.

\subsection{Integer Linear Programming}
\label{sub:ilp}

For our purposes, it will be useful to view an ILP instance as a set
of linear inequalities rather than using the constraint
matrix. Formally, let an ILP instance $I$ be a tuple $(\cF,\eval)$
where $\cF$ is a set of linear inequalities over variables
$X=\{x_1,\dots,x_n\}$ and $\eval$ is a linear function over $X$ of the
form $\eval(X)=s_1x_1+\dots+s_nx_n$. Each inequality
$A\in\cF$ ranges over variables $\var(A)$ is said to have arity $|\var(A)|=l$ 
and is assumed to be of the form $c_{A,1}x_{A,1}+c_{A,2}x_{A,2}+\dots + c_{A,l}x_{A,l} \leq b_A$; we also define $\var(I)=X$. We say that two constraints are equal if they range over the same variables with the same coefficients and have the same right-hand side.

For a set of variables $Y$, let $\cF(Y)$ denote
the subset of $\cF$ containing all inequalities $A\in \cF$ such that
$Y\cap \var(A)\neq \emptyset$. We will generally use the term \emph{coefficients} to refer to numbers that occur in the inequalities in $\cF$. 
In some cases, we will be dealing with certain selected ``named'' variables which will not be marked with subscripts to improve readability (e.g., $a$); there, we may use $s_a$ to denote the coefficient of $a$ in $\eval$, i.e., $s_a$ is shorthand for $s_j$ where $a=x_j$.

An assignment $\alpha$ is a mapping from $X$ to $\mathbb{Z}$. For an
assignment $\alpha$ and an inequality $A$ of arity $l$, we denote by $A(\alpha)$
the left-side value of $A$ obtained by applying $\alpha$, i.e.,
$A(\alpha)=c_{A,1}\alpha(x_{A,1})+c_{A,2}\alpha(x_{A,2})+\dots+
c_{A,l}\alpha(x_{A,l})$. Similarly, we let $\eval(\alpha)$ denote the value of the linear function $\eval$ after applying $\alpha$. 

An
assignment $\alpha$ is called \emph{feasible} if it satisfies every
$A\in \cF$, i.e., if $A(\alpha)\leq b_A$ for each $A\in
\cF$. Furthermore, $\alpha$ is called a \emph{solution} if the value
of $\eval(\alpha)$ is maximized over all feasible assignments; observe
that the existence of a feasible assignment does not guarantee the
existence of a solution (there may exist an infinite sequence of
feasible assignments $\alpha$ with increasing values of
$\eval(\alpha)$). Given an instance $I$, the task in the ILP problem
is to compute a solution for $I$ if one exists, and otherwise to
decide whether there exists a feasible assignment. On the other hand, 
the \probfont{ILP-feasibility} problem asks whether a given instance $I$ 
admits a feasible assignment (here, we may assume 
without loss of generality that all coefficients in $\eval$ are equal to $0$).

Given an ILP instance $I=(\cF,\eval)$, the primal graph $G_I$ of $I$
is the graph whose vertex set is the set $X$ of variables in $I$, and
two vertices $a,b$ are adjacent iff either there exists some $A\in
\cF$ containing both $a$ and $b$ or $a,b$ both occur in $\eval$ with
non-zero coefficients. 

\subsection{Parameterized Complexity}
\label{sub:pc}

In parameterized
algorithmics~\cite{CyganFKLMPPS15,FlumGrohe06,book/Niedermeier06,book/DowneyF13}
the runtime of an algorithm is studied with respect to a parameter
$k\in\nn$ and input size~$n$.
The basic idea is to find a parameter that describes the structure of
the instance such that the combinatorial explosion can be confined to
this parameter.
In this respect, the most favorable complexity class is \FPT (\textit{fixed-parameter tractable})
which contains all problems that can be decided by an algorithm
running in time $f(k)\cdot n^{\bigO{1}}$, where $f$ is a computable
function.
Algorithms with this running time are called \emph{fpt-algorithms}.

To obtain our lower bounds, we will need the notion of a
\emph{parameterized reduction} and the complexity class
\paraNP~\cite{book/DowneyF13}. Since we obtain all our lower bounds already for \probfont{ILP-feasibility}, we only need to consider these notions for decision problems; formally, a \textit{parameterized decision problem} is a subset of $\Sigma^*\times\nn$, where $\Sigma$ is the input alphabet. 

Let $L_1$ and $L_2$ be parameterized decision problems, with $L_1\subseteq \Sigma_1^*\times\nn$ and $L_2\subseteq \Sigma_2^*\times\nn$. A \textit{parameterized reduction} (or fpt-reduction) from $L_1$ to $L_2$ is a mapping $P:\Sigma_1^*\times\nn\rightarrow\Sigma_2^*\times\nn$ such that:
\begin{enumerate}
\itemsep-0.2em
\item $(x,k)\in L_1$ if and only if $P(x,k)\in L_2$;   
\item the mapping can be computed by an fpt-algorithm with respect to parameter $k$;   
\item there is a computable function $g$ such that $k'\leq g(k)$, where $(x',k')=P(x,k)$. \end{enumerate}

There is a variety of classes capturing \textit{parameterized intractability}.
For our results, we require only the class \paraNP, which is defined as the class of problems that are solvable by a nondeterministic Turing-machine in fpt-time.
We will make use of the characterization of \paraNP-hardness given by Flum and Grohe~\cite{FlumGrohe06}, 
Theorem 2.14: any parameterized (decision) problem that remains \NP-hard when the parameter is set to some constant is \paraNP-hard.
Showing \paraNP-hardness for a problem rules out the existence of an fpt-algorithm under the assumption that $\P\neq\NP$. In fact, it even allows us to rule out algorithms running in time $n^{f(k)}$ for any function $f$ (these are sometimes called \emph{XP} algorithms).

For our algorithms, we will use the following result as a
subroutine. Note that this is a streamlined version of the original
statement of the theorem, as used in the area of parameterized
algorithms~\cite{FellowsLokshtanovMisraRS08,GanianKimSzeider15}.

\begin{THE}[\cite{Lenstra83,Kannan87,FrankTardos87}]
  \label{thm:pilp}
  An ILP instance $I=(\cF,\eval)$ can be solved in time
  $\bigoh(p^{2.5p+o(p)}\cdot |I|)$, where $p=|\var(I)|$.
\end{THE}

\subsection{Treewidth and Treedepth}
\label{sub:tdtw}

Treewidth is the most prominent structural parameter and has been extensively studied in a number of fields. In order to define treewidth, we begin with the definition of its associated decomposition.
    A \emph{tree-decomposition}~$\mathcal{T}$ of a graph $G=(V,E)$ is a pair 
  $(T,\chi)$, where $T$ is a tree and $\chi$ is a function that assigns each 
  tree node $t$ a set $\chi(t) \subseteq V$ of vertices such that the following 
  conditions hold: 
\begin{itemize}
  \item[(P1)] For every vertex~$u \in V$, there is a tree node~$t$ such that 
    $u \in \chi(t)$.
  \item[(P2)] For every edge $\{u,v\} \in E(G)$ there is a tree node
    $t$ such that $u,v\in \chi(t)$.
  \item[(P3)] For every vertex $v \in V(G)$,
    the set of tree nodes $t$ with $v\in \chi(t)$ forms a subtree of~$T$.
    \end{itemize}
  The sets $\chi(t)$ are called \emph{bags} of the decomposition~$\mathcal{T}$ and $\chi(t)$ 
  is the bag associated with the tree node~$t$. 
  The \emph{width} of a tree-decomposition 
  $(T,\chi)$ is the size of a largest bag minus~$1$.  A tree-decomposition of
  minimum width is called \emph{optimal}.  The \emph{treewidth} of a graph $G$,
  denoted by $\textup{tw}(G)$, is the
  width of an optimal tree decomposition of~$G$.

Another important notion that we make use of extensively is that of treedepth. 
Treedepth is a structural parameter closely related to treewidth, and the structure of graphs
of bounded treedepth is well understood~\cite{NOdM12}.
A useful way of thinking about graphs of bounded treedepth is that they are
(sparse) graphs with no long paths.

We formalize a few notions needed to define treedepth.
A \emph{rooted forest} is a disjoint union of rooted trees. 
For a vertex~$x$ in a tree~$T$ 
of a rooted forest, the \emph{height} (or {\em depth})
of~$x$ in the forest is the number of vertices in the path from 
the root of~$T$ to~$x$. The \emph{height of a rooted forest} is the maximum height of a vertex of the forest. 
\begin{DEF}[Treedepth]\label{def:td}
  Let the \emph{closure} of a rooted forest~$\cal F$ be the graph
  $\clos({\cal F})=(V_c,E_c)$ with the vertex set 
  $V_c=\bigcup_{T \in \cal F} V(T)$ and the edge set
  $E_c=\{xy \colon \text{$x$ is an ancestor of $y$ in some $T\in\cal F$}\}$.
  A \emph{treedepth decomposition}
  of a graph $G$ is a rooted forest $\cal F$ such that $G \subseteq \clos(\cal F)$.
  The \emph{treedepth} $\td(G)$ of a graph~$G$ is the minimum height of
  any treedepth decomposition of $G$. 
\end{DEF}

\noindent We will later use $T_x$ to denote the vertex set of the subtree of $T$ rooted at a vertex $x$ of $T$. Similarly to treewidth, it is possible to determine the treedepth of a graph in FPT time.
\begin{PRO}[\cite{NOdM12}]\label{pro:compute-td}
  Given a graph $G$ with $n$ nodes and a constant $w$, it is possible to
  decide whether $G$ has treedepth at most $w$, and if so, to compute an
  optimal treedepth decomposition of $G$ in time $\bigO{n}$.
\end{PRO}

\noindent The following alternative (equivalent) characterization of treedepth will be useful later for ascertaining the exact treedepth in our reduction (specifically in Lemma~\ref{the:hard-treedepth}).
\begin{PRO}[\cite{NOdM12}]\label{prop:alternativetd}
Let $G_i$ be the connected components of $G$. Then
\[
\td(G)=
\begin{cases}
1, & \text{if } |V(G)|=1;\\
1+\min_{v\in V(G)} td(G-v), & \text{if $G$ is connected and } |V(G)|>1;\\
\max_{i} td(G_i), & \text{otherwise.}
\end{cases}
\]
\end{PRO}

\noindent We conclude with a few useful facts about treedepth. 

\begin{PRO}[\cite{NOdM12}]
\label{pro:tdfacts}
~
\begin{enumerate}
\item If a graph $G$ has no path of length~$d$, then $\td(G)\leq d$. 
\item If~$\td(G) \leq d$, then~$G$ has no path of length~$2^d$.
\item $\tw(G) \leq \td(G)$.  
\item If~$\td(G) \leq d$, then $\td(G')\leq d+1$ for any graph $G'$ obtained by adding one vertex into $G$.
\end{enumerate}

Within this manuscript, for an ILP instance $I$ we will use \emph{treewidth (treedepth) of $I$} as shorthand for the treewidth (treedepth) of the primal graph $G_I$ of $I$.
\end{PRO} 

\section{Exploiting Treedepth to Solve ILP}
\label{sec:using}

Our goal in this section is to show that ILP is fixed parameter
tractable when parameterized by the treedepth of the primal graph and
the maximum coefficient in any constraint. We begin by formalizing our
parameters. Given an ILP instance $I$, let $\td(I)$ be the treedepth
of $G_I$ and let $\ell(I)$ be the maximum absolute coefficient
which occurs in any inequality in $I$; to be more precise,
$\ell(I)=\max \SB |c_{A,j}|, |b_A|\SM A\in \cF, j\in \Nat \SE$. When
the instance $I$ is clear from the context, we will simply write
$\ell$ and $k=\td(I)$ for brevity. 
We will now state our main algorithmic result of this section.

\begin{THE}
  \label{thm:main}
  ILP is fixed-parameter tractable parameterized by
  $\ell$ and $k$
\end{THE}

The main idea behind our fixed-parameter algorithm for ILP is to show
that we can reduce the instance into an ``equivalent instance'' such
that the number of variables of the reduced instance can be bounded
by our parameters $\ell$ and $k$. We then apply Theorem~\ref{thm:pilp}
to solve the reduced instance. 

For the following considerations, we fix an ILP instance
$I=(\cF,\eval)$ of size $n$ along with a treedepth decomposition $T$
of $G_I$ with depth $k$. 
Given a variable set $Y$, the operation of
\emph{omitting} consists of deleting all inequalities containing at
least one variable in $Y$ and all variables in $Y$; formally, omitting
$Y$ from $I$ results in the instance $I'=(\cF',\eval')$ where
$\cF'=\cF \setminus \cF(Y)$ and $\eval'$ is obtained by removing all
variables in $Y$ from $\eval$.

The following notion of equivalence will be crucial for the proof of Theorem~\ref{thm:main}.
Let $x,y$ be two
variables that share a common parent in $T$, and recall that $T_x$ ($T_y$) denotes the vertex set of the subtree of $T$ rooted at $x$ ($y$). We say that $x$ are $y$ are
\emph{equivalent}, denoted $x\sim y$, if there exists a bijective
function $\delta_{x,y}:T_x\rightarrow T_y$ (called the \emph{renaming
  function}) such that $\delta_{x,y}(\cF(T_x))=\cF(T_y)$; here
$\delta_{x,y}(\cF(T_x))$ denotes the set of inequalities in $\cF(T_x)$
after the application of $\delta_{x,y}$ on each variable in $T_x$. In other words, $x\sim y$ means that there exists a way of ``renaming'' the variables in $T_y$ so that $\cF(T_y)$ becomes $\cF(T_x)$.

  It
is easy to verify that $\sim$ is indeed an equivalence relation.
Intuitively, the following lemma shows that if $x \sim y$ for two variables $x$ and
$y$ of $I$, then (up to renaming) the set of all feasible
assignments of the variables in $T_x$ is equal to the set of all feasible
assignments of the variables in $T_y$; it will be useful to recall the meaning of $s_a$ from Subsection~\ref{sub:ilp}.
\begin{LEM}
  \label{lem:pruning}
  Let $x,y$ be two variables of $I$ such that $x\sim y$ and
  $s_a=0$ for each $a\in T_x\cup T_y$. Let
  $I'=(\cF',\eval')$ be the instance obtained from $I$ by omitting
  $T_y$. Then there exists a solution $\alpha$ of $\var(I)$ of value
  $w=\eval(\alpha)$ if and only if there exists a solution $\alpha'$ of
  $\var(I')$ of value $w=\eval'(\alpha')$. Moreover, a solution $\alpha$
  can be computed from any solution $\alpha'$ in linear time if the
  renaming function $\delta_{x,y}$ is known.
\end{LEM}

\begin{proof}
  Let $\alpha$ be a solution of $\var(I)$ of value
  $w=\eval(\alpha)$. Since $\cF'\subseteq \cF$, it follows that setting
  $\alpha'$ to be a restriction of $\alpha$ to $\var(I)\setminus T_y$
  satisfies every inequality in $\cF'$. Since variables in $T_y$ do not
  contribute to $\eval$, it also follows that
  $\eval(\alpha)=\eval(\alpha')$.

  On the other hand, let $\alpha'$ be a solution of $\var(I')$ of
  value $w=\eval'(\alpha')$. Consider the assignment $\alpha$ obtained
  by extending $\alpha'$ to $T_y$ by reusing the assignments of $T_x$
  on $T_y$. Formally, for each $z\in T_y$ we set
  $\alpha(z)=\alpha'(\delta^{-1}_{x,y}(z))$ and for all other variables
  $w\in \var(I')$ we set $\alpha(w)=\alpha'(w)$. By assumption,
  $\alpha$ and $\alpha'$ must assign the same values to any variable
  $w$ such that $s_w\neq 0$, and hence
  $\eval(\alpha)=\eval(\alpha')$. To argue feasibility, first observe
  that any $A\in \cF'$ must be satisfied by $\alpha$ since $\alpha$
  and $\alpha'$ only differ on variables which do not occur in
  $I'$. Moreover, by definition of $\sim$ for each $A\in \cF\setminus
  \cF'=\cF(T_y)$ there exists an inequality $A'\in \cF'$ such that
  $\delta_{x,y}(A')=A$. In particular, this implies that
  $A(\alpha)=A'(\alpha)=A'(\alpha')$, and since $A'(\alpha')\leq
  b_{A'}=b_A$ we conclude that $A(\alpha)\leq b_A$. Consequently,
  $\alpha$ satisfies $A$. 

  The final claim of the lemma follows from the construction of $\alpha$ described above.
\end{proof}

\newcommand{\neqcl}{\textup{\#C}}

In the following let $z$ be a variable of $I$ at depth $k-i$ in $T$
for every $i$ with $1\leq i < k$ and let $Z$ be the set of all
children of $z$ in $T$. Moreover, let $m$ be the maximum size of any
subtree rooted at a child of $z$ in $T$, i.e., $m:=\max_{z' \in Z}|T_{z'}|$.
We will show next that the number of equivalence classes among the children
of $z$ can be bounded by the
function $\neqcl(\ell,k,i,m):=2^{(2\ell+1)^{k+1}\cdot  m^i}$.
Observe that this bound depends only on $\ell$, $k$, $m$, and $i$ 
and not on the size of $I$. 
\begin{LEM}
  \label{lem:subtrees}
  The equivalence relation $\sim$ has at most $\neqcl(\ell,k,i,m)$ equivalence
  classes over $Z$.
\end{LEM}

\begin{proof}
  Consider an element $a\in Z$. By construction of $G_I$, each inequality $A\in \cF(T_a)$ only
  contains at most $k-i$ variables outside of $T_a$ (specifically, the
  ancestors of $a$) and at most $i$ variables in $T_a$. Furthermore,
  $b_A$ and each coefficient of a variable in $A$ is an integer whose
  absolute value does not exceed $\ell$. From this it follows that
  there exists a finite number of inequalities which can occur in
  $\cF(T_a)$. Specifically, the number of distinct combinations of
  coefficients for all the variables in $A$ and for $b_A$ is
  $(2\ell+1)^{k+1}$, and the number of distinct choices of variables
  in $\var(A)\cap T_a$ is upper-bounded by ${m\choose i}$, and so we
  arrive at $|\cF(T_a)|\leq (2\ell+1)^{k+1}\cdot {m\choose i}\leq
  (2\ell+1)^{k+1}\cdot m^i$. 

  Consequently, the set of inequalities for each child $y\in Z$ of $z$ has bounded cardinality.
%
  We will use this to bound the number of equivalence classes in $\neqcl(\ell,k,i,m)$ 
  by observing that two elements are equivalent if and only if they occur in precisely
  the same sets of inequalities (up to renaming). To formalize this intuition,
  we need a formal way of canonically renaming all
  variables in the individual subtrees rooted in $Z$; without
  renaming, each $\cF(T_y)$ would span a distinct set of variables and
  hence it would not be possible to bound the set of all such inequalities. 
So, for each $y$ let $\delta_{y,x_0}$ be a bijective
  renaming function which renames all of the variables in $T_y$ to the
  variable set $\{x_0^1, x_0^2,\dots,x_0^{|T_y|}\}$ (in an arbitrary
  way). Now we can formally define $\Gamma_z=\SB \cF(T_{x_0})
  \SM\delta_{y,x_0}(\cF(T_y)), y\in Z \SE$, and observe that
  $\Gamma_z$ has cardinality at most $2^{(2\ell+1)^{k+1}\cdot
    m^i}=\neqcl(\ell,k,i,m)$.
  To conclude the proof, recall
  that if two variables $a,b$ satisfy
  $\cF(T_a)=\delta_{b,a}(\cF(T_b))$ for a bijective renaming function
  $\delta_{b,a}$, then $b\sim a$. Hence, the absolute bound on the
  cardinality of $\Gamma_z$ implies that $\sim$ has at most
  $\neqcl(\ell,k,i,m)$ equivalence classes over~$Z$.
\end{proof}
It follows from the above Lemma that if $z$ has more than
$\neqcl(\ell,k,i,m)$ children, then two of those must be
equivalent. The next lemma shows that it is also possible to find such
a pair of equivalent children efficiently.
\begin{LEM}
  \label{lem:findprune}
  Given a subset $Z'$ of $Z$ with $|Z'|=\neqcl(\ell,k,i,m)+1$, then in
  time $\bigoh(\neqcl(\ell,k,i,m)^2\cdot m!m)$ one can find two children
  $x$ and $y$ of $Z$ such that $x \sim y$ together with a renaming function
  $\delta_{x,y}$ which certifies this. 
\end{LEM}

\begin{proof}
  Consider the following algorithm $\mathbb{A}$. First, $\mathbb{A}$
  computes a subset $Z'$ consisting of exactly (arbitrarily chosen) $\neqcl(\ell,k,i,m)+1$
  children of $Z$. Then $\mathbb{A}$
  branches over all distinct pairs $x,y\in Z'$ in time at most
  $\bigoh(\neqcl(\ell,k,i,m)^2)$. Second, $\mathbb{A}$ branches over all of the at most $m!$
  bijective renaming functions $\delta_{x,y}$. Third, $\mathbb{A}$
  computes $\delta_{x,y}(\cF(T_x))$ and tests whether it is equal to
  $\cF(T_y)$ (which takes at most $\bigoh(m)$ time); if this is the case, then
  $\mathbb{A}$ terminates and outputs $x,y$ and $\delta_{x,y}$.

  We argue correctness. By Lemma~\ref{lem:subtrees} and due to the
  cardinality of $Z'$, there must exist $x,y\in Z'$ such that $x\sim
  y$. In particular, there must exist a renaming function
  $\delta_{x,y}$ such that $\delta_{x,y}(\cF(T_x))=\cF(T_y)$. But then
  $\mathbb{A}$ is guaranteed to find such $x,y,\delta_{x,y}$ since it
  performs an exhaustive search.
\end{proof}

Combining Lemma~\ref{lem:pruning} and Lemma~\ref{lem:findprune}, we arrive at the following corollary. 

\begin{COR}
  \label{cor:prune}
  If
  $|Z|>\neqcl(\ell,k,i,m)+1$, then in
  time $\bigoh(\neqcl(\ell,k,i,m)^2\cdot m!m)$ one can 
  compute a subinstance $I'=(\cF',\eval)$ of $I$ with strictly less
  variables and the following property: there exists a solution $\alpha$
  of $I$ of value $w=\eval(\alpha)$ if and only if there exists a
  solution $\alpha'$ of $I'$ of value $w$. Moreover, a solution $\alpha$
  can be computed from any solution $\alpha'$ in linear time.
\end{COR}

\begin{proof}
  In order to avoid having to consider all children of $z$, the
  algorithm first computes (an arbitrary) subset $Z'$ of $Z$ such that
  $|Z'|=\neqcl(\ell,k,i,m)+2$. Then to be able to apply
  Lemma~\ref{lem:findprune} without changing the set of solutions of
  $I$, the algorithm computes a subset $Z''$ of $Z'$ 
  such that $|Z''|=\neqcl(\ell,k,i,m)+1$ and for every $z' \in Z''$ it holds 
  that $s_{z''}=0$ for every $z'' \in T_{z'}$. Note that since there are
  at most $k$ variables of $I$ with non-zero coefficients in $\eval$
  and these variables form a clique in $G_I$, all of them occur only
  in a single branch of $T_z$. It follows that $Z''$ as specified
  above exists and it can be obtained from $Z'$ by removing the (at
  most one) element $z'$ in $Z'$ with $s_{z''}\neq 0$ for some $z'' \in T_{z'}$.
  Observe that this step of the algorithm takes time at most
  $\bigoh(m\cdot (\neqcl(\ell,k,i,m)+1))$.
  
  The algorithm then proceeds as
  follows. 
  It uses
  Lemma~\ref{lem:findprune} to find two variables $x,y\in Z''$ such that
  $x\sim y$ and computes $I'$ from $I$ by omitting $T_y$ from
  $I$. The running time of the algorithm follows
  from Lemma~\ref{lem:findprune} since the running times of the other
  steps of the algorithm are dominated by the application of
  Lemma~\ref{lem:findprune}.
  The corollary now follows from Lemma~\ref{lem:pruning} and Lemma~\ref{lem:findprune}, which certify that:
  \begin{itemize}
  \item there exists a solution $\alpha$
  of $I$ of value $w=\eval(\alpha)$ if and only if there exists a
  solution $\alpha'$ of $I'$ of value $w$, and
  \item a solution $\alpha$
  can be computed from any solution $\alpha'$ in linear time. \qedhere
  \end{itemize}
\end{proof}

Let $e_i$ and $d_i$ for every $i$ with $1 \leq i \leq k$ be defined
inductively by setting $e_k=1$, $d_k=0$, $d_i=\neqcl(\ell,k,i,s_{i+1})+1$, and 
$e_i=d_i e_{i+1}+1$.
The following Lemma shows that in time
$\bigoh(|I|d_1^2\cdot e_1!e_1)$ one can compute an ``equivalent''
subinstance $I'$ of $I$ containing at most $e_1$ variables. 
Informally, $e_i$ is an upper bound on the number of nodes in a subtree rooted at depth $i$ and $d_i$ is an upper bound on the number of children of a node at level $i$ in $I'$. 

\begin{LEM}
  \label{lem:fullprune}
  There exists an algorithm that takes as input $I$ and $T$, runs in
  time $\bigO{|I|d_1^2\cdot e_1!e_1}$ and outputs
  an ILP instance $I'$ containing at most
  $e_1$ variables with the following
  property: there exists a solution $\alpha$ of $I$ of value
  $w=\eval(\alpha)$ if and only if there exists a solution $\alpha'$ of
  $I'$ of value $w=\eval'(\alpha')$. 
  Moreover, a solution $\alpha$ can be computed from any solution $\alpha'$ in linear time.
\end{LEM}

\begin{proof}
  The algorithm exhaustively applies Corollary~\ref{cor:prune} to every
  variable of $T$ in a bottom-up manner, i.e., it starts by applying the
  corollary exhaustively to all variables at depth $k-1$ and then
  proceeds up the levels of $T$ until it reaches depth $1$. Let $T'$
  be the subtree of $T$ obtained after the exhaustive application of 
  Corollary~\ref{cor:prune} to $T$.

  We will
  first show that if $x$ is a variable at depth $i$ of $T'$,
  then $x$ has at most $d_i$ children and $|T_x'|\leq e_i$. 
  We will show the claim by induction on the depth $i$ starting from
  depth $k$. Because all variables $x$ of $T$ at level $k$ are leaves,
  it holds that $x$ has $0=d_k$ children in $T'$ and $|T_x'|=1\leq e_k$,
  showing the start of the induction. Now let $x$ be a variable at
  depth $i$ of $T'$ and let $y$ be a child of $x$ in $T'$. It follows
  from the induction hypothesis that $|T_y'|\leq e_{i+1}$. Moreover,
  using Corollary~\ref{cor:prune}, we obtain that 
  $x$ has at most
  $\tilde \neqcl(\ell,k,i,e_{i+1})+1=d_i$ children in $T'$
  and thus $|T_x'|\leq d_i e_{i+1}+1=e_i$, as required.

  The running time of the algorithm now follows from the observation
  that (because every application of Corollary~\ref{cor:prune} removes
  at least one variable of $I$) Corollary~\ref{cor:prune} is applied
  at most $|I|$ times and moreover the maximum running time of any
  call to Corollary~\ref{cor:prune} is at most $\bigoh(d_1^2\cdot e_1!e_1)$.
  Correctness and the fact that $\alpha$ can be computed from
  $\alpha'$ follow from Corollary~\ref{cor:prune}; more specifically, we extend $\alpha'$ into $\alpha$ by assigning pruned variables in the same way as their equivalent counterparts.
\end{proof}

\begin{proof}[Proof of Theorem~\ref{thm:main}]
  The algorithm proceeds in three steps. First, it applies
  Lemma~\ref{lem:fullprune} to reduce the instance $I$ into an
  ``equivalent'' instance $I'$ containing at most
  $e_1$ variables in time $\bigO{|I|d_1^2\cdot e_1!e_1}$; in particular, a
  solution $\alpha$ of $I$ can be computed in linear time from a
  solution $\alpha'$ of $I'$. Second, it uses Theorem~\ref{thm:pilp} to
  compute a solution $\alpha'$ of $I'$ in time at most
  $\bigoh(e_1^{2.5e_1+o(e_1)}\cdot |I'|)$;
  because $e_1$ and $d_1$ are bounded by our
  parameters, the whole algorithm runs in FPT time. Third, it transforms the
  solution $\alpha'$ into a solution $\alpha$ of $I$. 
  Correctness follows from Lemma~\ref{lem:fullprune} and Theorem~\ref{thm:pilp}.
\end{proof}

\section{Lower Bounds and Hardness}
\label{sec:hardness}

In this section we will complement our algorithmic results 
by providing matching hardness results. Namely, we will show that already the 
\probfont{ILP-feasibility} problem
is \NP{}-hard on graphs of bounded treedepth and also \NP{}-hard on graphs of bounded treewidth and bounded maximum coefficient\footnote{Unless explicitly mentioned otherwise, all the presented \NP{}-hardness results hold in the strong sense, i.e., when the input is encoded in unary.}.

We begin by noting that \probfont{ILP-feasibility} remains \NP{}\hy 
hard even if the maximum absolute value of any coefficient 
is at most one.
This follows, e.g., by enhancing the standard reduction from the decision version of 
\probfont{Vertex Cover} (given a graph $G$ and a bound $\nu$, does $G$ admit a vertex cover of size at most $\nu$?) to \probfont{ILP-feasibility} as follows:
\begin{itemize}
\item add variables $x_1,\dots, x_\nu$ and force each of them to be $1$,
\item set $x=\sum_{i\in [\nu]}x_i$,
\item add a constraint requiring that the sum of all variables which represent
vertices of $G$ is at most $x$.
\end{itemize}
\begin{OBS}\label{obs:hard-coef}
  ILP-feasibility is \NP{}\hy hard even on instances with a maximum
  absolute value of every coefficient of $1$.
\end{OBS}

To simplify the constructions in the hardness proofs, we will often
talk about constraints as equalities instead of inequalities. Clearly,
every equality can be written in terms of two inequalities.

\begin{THE}\label{the:hard-treedepth}
  \probfont{ILP-feasibility} is \NP{}\hy hard
even on instances of bounded treedepth.
\end{THE}
\begin{proof}
  We will show the theorem by a polynomial-time reduction from the
  well-known \NP-hard \probfont{$3$-Colorability}
  problem~\cite{GareyJohnson79}: given a graph, decide whether the vertices of $G$
  can be colored with three colors such that no two adjacent vertices
  of $G$ share the same color. 
  
  The main idea behind the reduction is
  to represent a $3$-partition of the vertex set of $G$ (which in
  turn represents a $3$-coloring of $G$) by the domain values of
  three ``global'' variables. The value of each of these global variables
  will represent a subset of vertices of $G$ that will be colored using
  the same color. To represent a subset of the vertices of $G$ in
  terms of domain values of the global variables, we will represent
  every vertex of $G$ with a unique prime number and a subset by the
  value obtained from the multiplication of all prime numbers of
  vertices contained in the subset. To ensure that the subsets
  represented by the global variables correspond to a valid
  $3$-partition of $G$ we will introduce constraints which ensure that:
  \begin{itemize}
  \item[C1]  For every prime number representing some vertex of $G$
    exactly one of the global variables is divisible by that prime
    number. This ensures that every vertex of $G$ is assigned to
    exactly one color class.
  \item[C2] For every edge $\{u,v\}$ of $G$ it holds that no global
    variable is divisible by the prime numbers
    representing $u$ and $v$ at the same time. This ensures that no
    two adjacent vertices of $G$ are assigned to the same color class.
  \end{itemize}
  Thus let $G$ be the given instance of \probfont{$3$-Coloring} and
  assume that the vertices of $G$ are uniquely identified as elements of $\{1,\dots,|V(G)|\}$.
  In the following we denote by $p(i)$ the
  $i$-th prime number for any positive integer $i$, where $p(1)=2$.
  We construct an instance $I$ of
  \probfont{ILP-feasibility} in polynomial-time with treedepth at
  most $8$ and coefficients bounded by a polynomial in $V(G)$ such that $G$ has a
  $3$\hy coloring if and only if $I$ has a feasible assignment.
This instance $I$ has the following variables:
  \begin{itemize}
  \item The \emph{global variables} $g_1$, $g_2$, and $g_3$ with
    an arbitrary positive domain, whose
    values will represent a valid $3$-Partioning of $V(G)$.
  \item For every $i$ and $j$ with $1 \leq i \leq |V(G)|$ and $1 \leq
    j \leq 3$, the variables $m_{i,j}$ (with an arbitrary non-negative domain),
    $r_{i,j}$ (with domain between $0$ and $p(i)-1$), and
    $u_{i,j}$ (with binary domain). These variables are used to secure
    condition C1.
  \item For every $e \in E(G)$, $v \in e$, and $j$ with $1 \leq
    j \leq 3$, the variables $m_{e,v,j}$ (with an arbitrary non-negative domain),
    $r_{e,v,j}$ (with domain between $0$ and $p(v)-1$), and
    $u_{e,v,j}$ (with binary domain). These variables are used to secure condition C2.
  \end{itemize}

  \noindent $I$ has the following constraints (in the following let $\alpha$ be
  any feasible assignment of $I$):
  \begin{itemize}
  \item Constraints that restrict the domains of all variables as
    specified above, i.e.:
    \begin{itemize}
    \item  for every $i$ and $j$ with $1 \leq i \leq
      |V(G)|$ and $1 \leq j \leq 3$, the constraints $g_j \geq 0$,
      $m_{i,j} \geq 0$, $0 \leq r_{i,j}\leq p(i)-1$, and $0\leq u_{i,j}
      \leq 1$.
    \item for every $e\in E(G)$, $v \in e$, and $j$ with $1 \leq j
      \leq 3$, the constraints $m_{e,v,j}\geq 0$, $0\leq r_{e,v,j}\leq
      p(v)-1$, and $0 \leq u_{e,v,j}\leq 1$.
    \end{itemize}
  \item The following constraints, introduced for each $1 \leq i
      \leq |V(G)|$ and $1 \leq j \leq 3$, together guarantee that condition C1 holds:
    \begin{itemize}
    \item Constraints that ensure that $\alpha(r_{i,j})$
      is equal to the remainder of $\alpha(g_j)$ divided by $p(i)$, i.e., the constraint 
      $g_j=p(i)m_{i,j}+r_{i,j}$.
    \item
      Constraints that ensure that $\alpha(u_{i,j})=0$ 
      if and only if $\alpha(r_{i,j})=0$, i.e., the
      constraints $u_{i,j}\leq r_{i,j}$ and $r_{i,j} \leq
      (p(i)-1)u_{i,j}$.
      Note that together the above constraints now ensure that
      $\alpha(u_{i,j})=0$ if and only if $g_j$ is divisible by $p(i)$.
    \item 
      Constraints that ensure that exactly one of $\alpha(u_{i,1})$,
      $\alpha(u_{i,2})$, and $\alpha(u_{i,3})$ is equal to $0$, i.e.,
      the constraints $2 \leq u_{i,1}+u_{i,2}+u_{i,3}\leq 2$.
      Note that together all the above constraints now ensure condition
      C1 holds.
    \end{itemize}
  \item The following constraints, introduced for each $1 \leq j \leq 3$, together guarantee that condition C2 holds:
    \begin{itemize}
    \item Constraints that ensure that for every $e \in E(G)$ and $v \in e$,
      it holds that $\alpha(r_{e,v,j})$
      is equal to the remainder of $g_j$ divided by $p(v)$, i.e., the constraint 
      $g_j=p(i)m_{e,v,j}+r_{e,v,j}$.
    \item
      Constraints that ensure that for every $e \in E(G)$, $v \in e$,
      and $j$ with $1 \leq j \leq 3$ it holds that
      $\alpha(u_{e,v,j})=0$ if and only if $\alpha(r_{e,v,j})=0$, i.e., the
      constraints $u_{e,v,j}\leq r_{e,v,j}$ and $r_{e,v,j} \leq
      p(v)u_{e,v,j}$.
      Note that together the above constraints now ensure that
      $\alpha(u_{e,v,j})=0$ if and only if $g_j$ is divisible by $p(v)$.
    \item 
      Constraints that ensure that for every $e=\{v,w\} \in E(G)$
      and $j$ with $1 \leq j \leq 3$ it holds that at least one of
      $\alpha(u_{e,w,j})$ and $\alpha(u_{e,v,j})$ is non-zero, i.e., the
      constraint $u_{e,u,j}+u_{e,v,j}\geq 1$. Note that together with
      all of the above constraints this now ensures condition C2.
    \end{itemize}
  \end{itemize}
  This completes the construction of $I$ and the largest coefficient used in $I$ is $p(|V(G)|)$.
  It is well-known that $p(i)$ is upper-bounded
  by $\bigoh(i\log i)$ due to the Prime Number Theorem, and so 
  this in particular implies that the numbers which occur in $I$ are
  bounded by a polynomial in $|V(G)|$. Hence $I$ can be constructed in polynomial time.
  
  Following the construction and explanations provided above, it is not difficult to see that $I$ has a
  feasible assignment if and only if $G$ has a $3$-coloring. Indeed, for any $3$-coloring of $G$, one can construct a feasible assignment of $I$ by computing the prime-number encoding for the vertex sets that receive colors $1,2,3$ and assign these three numbers to $g_1,g_2,g_3$, respectively. Such an assignment allows us to straightforwardly satisfy the constraints ensuring C1 holds (since each prime occurs in exactly one global constraint), the constraints ensuring C2 holds (since each edge is incident to at most one of each color) while maintaining the domain bounds. 
  
  On the other hand, for any feasible assignment $\alpha$, clearly each of $\alpha(g_1), \alpha(g_2), \alpha(g_3)$ will be divisible by some subset of prime numbers between $2$ and $p(|V(G)|)$. In particular, since $\alpha$ is feasible it follows from the construction of our first group of constraints that each prime between $2$ and $p(|V(G)|)$ divides precisely one of $\alpha(g_1), \alpha(g_2), \alpha(g_3)$, and so this uniquely encodes a corresponding candidate $3$-coloring for the vertices of the graph. Finally, since $\alpha$ also satisfies the second group of constraints, this candidate $3$-coloring must have the property that each edge is incident to exactly $2$ colors, and so it is in fact a valid $3$-coloring.
  


  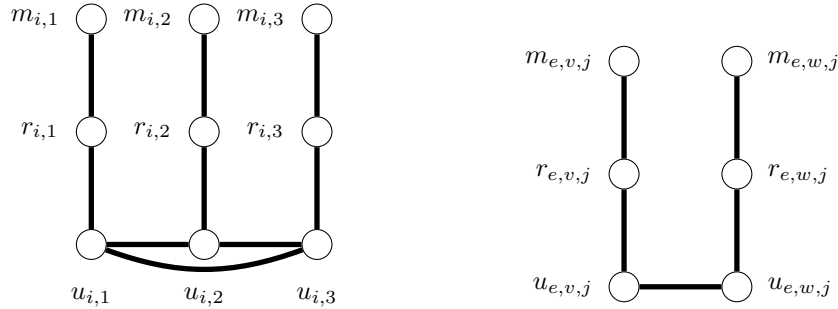
\begin{figure}[tb]
    \centering
    \begin{tikzpicture}[]
      \tikzstyle{every node}=[circle, inner sep=4pt,draw, node distance=1.5cm]
      \tikzstyle{every edge}=[draw, line width=2pt]
      \draw
      node[label=left:$m_{i,1}$] (m1) {} 
      node[right of=m1, label=left:$m_{i,2}$] (m2) {}
      node[right of=m2, label=left:$m_{i,3}$] (m3) {}

      node[below of=m1, label=left:$r_{i,1}$] (r1) {}
      node[right of=r1, label=left:$r_{i,2}$] (r2) {}
      node[right of=r2, label=left:$r_{i,3}$] (r3) {}

      node[below of=r1, label=below:$u_{i,1}$] (u1) {}
      node[right of=u1, label=below:$u_{i,2}$] (u2) {}
      node[right of=u2, label=below:$u_{i,3}$] (u3) {} 
      ;
      
      \draw
      (m1) edge (r1)
      (r1) edge (u1)
      (m2) edge (r2)
      (r2) edge (u2)
      (m3) edge (r3)
      (r3) edge (u3)
      (u1) edge (u2)
      (u2) edge (u3)
      (u1) edge[bend right=20] (u3)
      ;
    \end{tikzpicture}\hspace{20mm}
    \begin{tikzpicture}
      \tikzstyle{every node}=[circle, inner sep=4pt,draw, node distance=1.5cm]
      \tikzstyle{every edge}=[draw, line width=2pt]
      \draw
      node[label=left:$m_{e,v,j}$] (mv) {} 
      node[right of=mv, label=right:$m_{e,w,j}$] (mw) {}

      node[below of=mv, label=left:$r_{e,v,j}$] (rv) {}
      node[right of=rv, label=right:$r_{e,w,j}$] (rw) {}

      node[below of=rv, label=left:$u_{e,v,j}$] (uv) {}
      node[right of=uv, label=right:$u_{e,w,j}$] (uw) {}
      ;
      
      \draw
      (mv) edge (rv)
      (rv) edge (uv)
      (mw) edge (rw)
      (rw) edge (uw)
      (uv) edge (uw)
      ;
    \end{tikzpicture}
  
  \caption{Illustration of a vertex-type component (left) and an edge-type component (right) in the graph $G_I\setminus\{g_1,g_2,g_3\}$.}
  \label{fig:hard-td-vert-edge}
\end{figure}

  It remains to show that the treedepth of $I$ is at most $8$.
  We will show this by using the characterization of treedepth given in Proposition~\ref{prop:alternativetd}.
  We first observe that the graph $G_I \setminus
  \{g_1,g_2,g_3\}$ consists of the following components:
  \begin{itemize}
  \item for every $i$ with $1 \leq i \leq |V(G)|$, 
    one component on the vertices $m_{i,1},\dotsc, m_{i,3}$,
    $r_{i,1},r_{i,2},r_{i,3}$, $u_{i,1},u_{i,2},u_{i,3}$. 
    Note that all of these components are isomorphic to each other and
    we will therefore in the following refer to these components as \emph{vertex-type
      components}.
  \item for every $e=\{w,v\} \in E(G)$ and $j$ with $1 \leq j \leq 3$,
    one component on the vertices $m_{e,w,j}$, $m_{e,v,j}$,
    $r_{e,w,j}$, $r_{e,v,j}$, $u_{e,w,j}$, and $u_{e,v,j}$.
    Note that all of these components are isomorphic to each other and
    we will therefore in the following refer to these components as \emph{edge-type
      components}.
  \end{itemize}
  The two types of components are illustrated in Figure~\ref{fig:hard-td-vert-edge}.
  We will show next that any vertex-type component has treedepth at
  most $5$ and every edge-type component has treedepth at most
  $4$. This would then imply that $G_I$ has treedepth at most $8$ (since it suffices to remove the vertices $\{g_1,g_2,g_3\}$ in order to decompose the graph into these components).
  Hence let $i$ with $1 \leq i \leq |V(G)|$ and consider the
  vertex-type component $C_i$ on the vertices $m_{i,1},m_{i,2}, m_{i,3}$,
  $r_{i,1},r_{i,2},r_{i,3}$, $u_{i,1},u_{i,2},u_{i,3}$. Note that
  $C_i \setminus \{u_{i,1},u_{i,2},u_{i,3}\}$ consists of one component
  for every $j$ with $1 \leq j \leq 3$ that contains the vertices
  $m_{i,j}$ and $r_{i,j}$. Clearly each of these three components has
  treedepth at most $2$ and hence the treedepth of $C_i$ is at most
  $2+3=5$, as required.

  In order to show that every edge-type component has treedepth at
  most $4$, consider an edge $e=\{w,v\} \in E(G)$ and some $j$ satisfying $1\leq
  j \leq 3$. Let
  $C_{e,j}$ be the edge-type component consisting of the vertices 
  $m_{e,w,j}$, $m_{e,v,j}$, $r_{e,w,j}$, $r_{e,v,j}$, $u_{e,w,j}$, and
  $u_{e,v,j}$. Note that $C_{e,j} \setminus \{u_{e,w,j},u_{e,v,j}\}$
  consists of two components, one containing the vertices $m_{e,w,j}$
  and $r_{e,u,j}$ and one containing the vertices $m_{e,v,j}$ and
  $r_{e,v,j}$. Clearly, each of these two components has treedepth at
  most $2$ and hence the treedepth of $C_{e,j}$ is at most $2+2=4$,
  as required.
\end{proof}

The next theorem shows that \probfont{ILP-feasibility} is \paraNP{}\hy
hard parameterized by both treewidth and the maximum absolute value of
any number in the instance; observe that since we are bounding all
numbers in the instance, the theorem in particular implies
\NP{}\hy hardness. We note that the idea to reduce from \textsc{Subset Sum} was inspired by previous work of Jansen and Kratsch~\cite{JansenKratsch15esa}.
\begin{THE}\label{the:para-hard-tw-coeff}
  \probfont{ILP-feasibility} is \NP{}\hy hard even
  on instances with treewidth at most two and where the maximum absolute value of any
  coefficient is at most one.
\end{THE}
\begin{proof}
  We show the result by a polynomial reduction from the
  \probfont{Subset Sum} problem, which is well-known to be weakly \NP{}\hy
  complete. 
    \pbDef{\probfont{Subset Sum}}
  {A set $Q:=\{q_1,\dotsc,q_n\}$ of integers and an integer $r$.}
  {Is there a subset $Q' \subseteq Q$ such that $\sum_{q'\in Q'}q'=r$?}
  Let $I:=(Q,r)$ with $Q:=\{q_1,\dotsc,q_n\}$ be an instance of
  \probfont{Subset Sum}, which we assume to be given in binary encoding.
  We will construct an instance $I'$ of \probfont{ILP-feasibility}
  equivalent to $I$ in polynomial-time 
  (with respect to the input size of $I$)
  with treewidth at most $2$ that
  uses only $-1$, $0$, and $1$ as coefficients.   Crucial to our construction are the following auxiliary ILP instances.
  
  \begin{figure}[tb]
    \centering
    \begin{tikzpicture}[node distance=1.7cm]
      \tikzstyle{ver}=[draw, circle, inner sep=4pt]
      \tikzstyle{every edge}=[draw,line width=2pt]
      \tikzstyle{nd}=[circle, inner sep=1pt,draw]

      \draw

      node[ver,label=above:$x$] (x) {}
      node[ver, right of=x, label=above:$z_{0}$] (zi0) {}
      node[ver, right of=zi0, label=above:$z_{1}$] (zi1) {}
      node[ver, right of=zi1, label=above:$z_{2}$] (zi2) {}

      node[nd, node distance=0.5cm, right of=zi2] (dz1) {}
      node[nd, node distance=0.25cm, right of=dz1] (dz2) {}
      node[nd, node distance=0.25cm, right of=dz2] (dz3) {}
      
      node[ver, node distance=0.5cm,right of=dz3, label=above:$z_{m-2}$] (zi3) {}
      node[ver, right of=zi3, label=above:$z_{m-1}$] (zi4) {}
      node[ver, right of=zi4, label=above:$z_{m}$] (zi5) {}
      node[ver, right of=zi5, label=above:$y$] (y) {}

      node[ver,below of=zi0,label=below right:$h_{0}$] (bi0) {}
      node[ver, right of=bi0, label=below right:$h_{1}$] (bi1) {}
      node[ver, right of=bi1, label=below right:$h_{2}$] (bi2) {}

      node[nd, node distance=0.5cm, right of=bi2] (db1) {}
      node[nd, node distance=0.25cm, right of=db1] (db2) {}
      node[nd, node distance=0.25cm, right of=db2] (db3) {}

      node[ver, node distance=0.5cm, right of=db3, label=below right:$h_{m-2}$] (bi3) {}
      node[ver, right of=bi3, label=below right:$h_{m-1}$] (bi4) {}
      node[ver, right of=bi4, label=below right:$h_{m}$] (bi5) {}

      node[ver,below of=bi0,label=below:$h_{0}'$] (bi0p) {}
      node[ver, right of=bi0p, label=below:$h_{1}'$] (bi1p) {}
      node[ver, right of=bi1p, label=below:$h_{2}'$] (bi2p) {}

      node[nd, node distance=0.5cm, right of=bi2p] (db1p) {}
      node[nd, node distance=0.25cm, right of=db1p] (db2p) {}
      node[nd, node distance=0.25cm, right of=db2p] (db3p) {}

      node[ver, node distance=0.5cm, right of=db3p,
      label=below:$h_{m-2}'$] (bi3p) {}
      node[ver, right of=bi3p, label=below:$h_{m-1}'$] (bi4p) {}
            
      ;
      
      \draw
      (x) edge (zi0)
      (x) edge (bi0)
      (bi0) edge (bi1)
      (bi1) edge (bi2)
      (bi3) edge (bi4)
      (bi4) edge (bi5)
      
      (zi0) edge (zi1)
      (zi1) edge (zi2)
      (zi3) edge (zi4)
      (zi4) edge (zi5)
      
      (bi1) edge (zi0)
      (bi2) edge (zi1)
      (bi4) edge (zi3)
      (bi5) edge (zi4)
      
      (bi0) edge (zi0)
      (bi1) edge (zi1)
      (bi2) edge (zi2)
      (bi3) edge (zi3)
      (bi4) edge (zi4)
      (bi5) edge (zi5)

      (bi0) edge (bi0p)
      (bi1) edge (bi1p)
      (bi2) edge (bi2p)
      (bi3) edge (bi3p)
      (bi4) edge (bi4p)

      (bi0p) edge (bi1)
      (bi1p) edge (bi2)
      (bi3p) edge (bi4)
      (bi4p) edge (bi5)
      
      (y) edge (zi5)
      ;
    \end{tikzpicture}
    \caption{Illustration of the primal graph of the instance $I(q,x,y)$.}
    \label{fig:hard-tw-coeff}
  \end{figure}
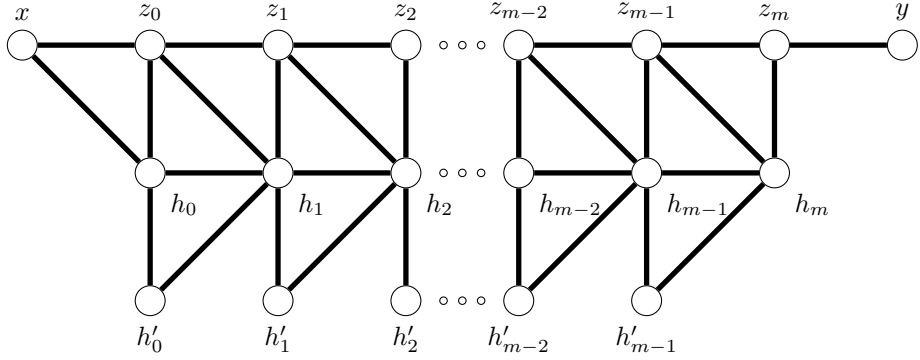

    \begin{figure}[tb]
    \centering
    \begin{tikzpicture}[node distance=1.5cm]
      \tikzstyle{box}=[draw, rectangle, inner sep=3pt, minimum height=1cm]
      \tikzstyle{ver}=[draw, circle, fill,inner sep=3pt]
      \tikzstyle{every edge}=[draw,line width=2pt]
      \tikzstyle{nd}=[circle, inner sep=1pt,draw]

      \draw
      node[box] (b1) {$I(q_1,y_1)$}
      (node cs:name=b1, anchor=east) +(0.5cm,0) node[ver, label=above:$y_1$] (y1) {}
      (y1) edge (node cs:name=b1, anchor=east)
      
      (y1) +(0.5cm,0) node[box, anchor=west] (b2) {$I(q_2,y_1,y_2)$}
      (node cs:name=b2, anchor=east) +(0.5cm,0) node[ver, label=above:$y_2$] (y2) {}
      (y1) edge (node cs:name=b2, anchor=west)
      (y2) edge (node cs:name=b2, anchor=east)

      node[nd, node distance=0.5cm, right of=y2] (d1) {}
      node[nd, node distance=0.25cm, right of=d1] (d2) {}
      node[nd, node distance=0.25cm, right of=d2] (d3) {}

      (d3) +(1cm,0) node[box, anchor=west] (b3)
      {$I(q_n,y_{n-1},y_n)$}
      (node cs:name=b3, anchor=west) +(-0.5cm,0) node[ver, label=above:$y_{n-1}$] (y3) {}
      (node cs:name=b3, anchor=east) +(0.5cm,0) node[ver, label=above:$y_n$] (y4) {}
      (y3) edge (node cs:name=b3, anchor=west)
      (y4) edge (node cs:name=b3, anchor=east)

      (y4) +(0.5cm,0) node[box, anchor=west] (b4)
      {$I_C(r,y_{n})$}
      (y4) edge (node cs:name=b4, anchor=west)

      ;
      
      \draw
      ;
    \end{tikzpicture}
    \caption{Illustration of the ILP instance $I'$.}
    \label{fig:hard-tw-coeff-2}
  \end{figure}
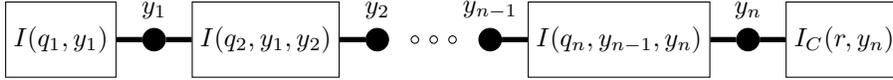
  \begin{CLM}\label{clm:hard-tw-coeff}
    For every $q \in \Nat$ and any two variables $x$ and $y$ there is an ILP instance $I(q,x,y)$
    satisfying the following conditions:
    \begin{itemize}
    \item[(P1)] $I(q,x,y)$ has at most $\bigoh(\log q)$ variables and
      constraints,
    \item[(P2)] the maximum absolute value of any coefficient in
      $I(q,x,y)$ is at most one,
    \item[(P3)] the treewidth of $I(q,x,y)$ is at most two and 
    \item[(P4)] for every feasible assignment
      $\alpha$ of $I(q,x,y)$, it holds that $\alpha(y)\in \{\alpha(x),\alpha(x)+q\}$.
    \end{itemize}
    Moreover, there are ILP instances $I(q,y)$ and $I_C(q,y)$
    satisfying (P1)--(P3) and additionally:
    \begin{itemize}
    \item $\alpha(y) \in \{0,q\}$ for $I(q,y)$, and
    \item $\alpha(y) =q$ for $I_C(q,y)$.
    \end{itemize}
  \end{CLM}
  \begin{proof}
    For an integer $q$, let $B(q)$ be the set of indices of all bits
    that are equal to one in the binary representation of $q$, i.e., we 
    have $q=\sum_{j \in B(q)}2^{j}$. Moreover, 
    let $m=b_{\max}(q)$ be the largest index in $B(q)$.

    We construct the ILP instance $I(q,x,y)$ as follows.
    We first introduce $m+1$ variables $h_0,\dotsc,h_m$ together
    with $m$ variables $h_0',\dotsc,h_{m-1}'$ and add the following
    constraints: $0\leq h_0\leq 1$, and
    for every $i$ with $0 \leq i < m$ we set
    $h_i'=h_i$ and $h_{i+1}=h_i+h_i'$.
    Observe that the above constraints ensure that
    $\alpha(h_{i})$ is equal to $2^i\alpha(h_0)$ for every $i$ with $0
    \leq i \leq m$ and every feasible assignment $\alpha$.
    We also introduce the new auxiliary variables
    $z_{0},\dotsc,z_{m}$
    together with the following constraints:
    \begin{itemize}
    \item If $0 \in B(q)$ then we add the constraint $z_{0}=h_{0}+x$, and
      otherwise we add the constraint $z_{0}=x$.
    \item For
      every $i$ with $0\leq i<m$, if $i+1 \in B(q)$ then we add the constraint
      $z_{i+1}=h_{i+1}+z_{i}$ and otherwise the constraint 
      $z_{i+1}=z_{i}$.
    \end{itemize}
    Observe that these constraints ensure that
    $\alpha(z_{i})$ is equal to $\alpha(x)+\sum_{j \in B(q)\land j \leq i}\alpha(h_{j})$
    for every $i$ with $0\leq i \leq m$ and any feasible
    assignment $\alpha$.
    Finally we introduce the constraint $y=z_m$.
    This concludes the construction of $I(q,x,y)$.
    By construction $I(q,x,y)$ satisfies (P1) and (P2).
    Moreover, because $\alpha(y)=\alpha(z_m)$ is equal
    to $q\alpha(h_0)+\alpha(x)$ for any
    feasible assignment $\alpha$ and since $\alpha(h_0)\in \{0,1\}$,
    we obtain that $\alpha(y) \in \{\alpha(x),\alpha(x)+q\}$ showing that $I(q,x,y)$
    satisfies (P4). Finally, with the help of Figure~\ref{fig:hard-tw-coeff},
    it is straightforward to verify that
    $I(q,x,y)$ has treewidth at most two.

    The ILP instance $I(q,y)$ can now be obtained from $I(q,x,y)$ by
    removing the variable $x$. Moreover, the ILP instance $I_C(q,y)$
    can now be obtained from $I(q,y)$ by replacing the
    constraints $0 \leq h_0 \leq 1$ with the constraint $h_0=1$.
  \end{proof}

  We now obtain $I'$ as the (non-disjoint) union of the instances
  $I(q_1,y_1)$, $I(q_{i},y_{i-1},y_{i})$ for every $i$ with $1< i \leq
  n$, and the instance $I_C(r,y_n)$ (see Figure~\ref{fig:hard-tw-coeff-2} for an
  illustration of $I'$). The size of each of these
  $n+1$ instances is bounded by $\bigoh(\log m)$, where $m$ is the maximum
  of $\{q_1,\dotsc,q_n,r\}$, and it can be verified that each of these instances can be constructed in time $\bigoh(\log m)$. Hence the construction of $I'$ from $I$ can be completed in polynomial time (with respect to the size of the binary encoding of $I$).   We also observe that the maximum
  absolute value of any coefficient in $I'$ is at most $1$. Finally,
  because $I'$ is a simple concatenation of ILP instances with
  treewidth at most $2$, it is straightforward to verify that $I'$ has
  treewidth at most $2$.  
\end{proof}

\section{Concluding Notes}
\label{sec:con}

We presented new results that add to the complexity landscape for
ILP w.r.t.\ structural parameterizations of the constraint matrix. 
Our main algorithmic result pushes the
frontiers of tractability for ILP instances and will hopefully
serve as a precursor for the study of further structural
parameterizations for ILP. We note that the running time of 
the presented algorithm has a highly nontrivial dependence on the treedepth
of the ILP instance, and hence the algorithm is unlikely to 
outperform dedicated solvers in practical settings.

The provided results draw
an initial complexity landscape for ILP
w.r.t.\ the most prominent decompositional width parameters. 
However, other approaches exploiting the structural properties of ILP
instances still remain unexplored and represent interesting directions for future research.
For instance, an adaptation of \emph{backdoors}~\cite{GaspersSzeider12} to the ILP setting could lead to highly relevant algorithmic results.


\paragraph{Acknowledgments} The authors acknowledge support by the Austrian Science Fund (FWF, project P26696). Robert Ganian is also affiliated with FI MU, Brno, Czech Republic. The authors also thank the anonymous reviewers for many insightful suggestions and comments---and in particular for helping improve Theorem~\ref{the:para-hard-tw-coeff}.

\bibliographystyle{aaai}
\bibliography{literature}

\end{document}